\documentclass[11pt,twoside]{article} 
\usepackage{euler,acta-info}
\usepackage[latin2]{inputenc}
\title{%
Testing of random matrices
}

\newcommand{\vol}{\textbf{3,} 1 (2011) 99--126}   
\newcommand{\amss}{\amssubj{%
68M20, 05B15
}}   
\newcommand{\CRs}{\CRsubj{%
G.2.2
}}   
\newcommand{\keyw}{\keywords{%
random sequences, analysis of algorithms, Latin squares, Sudoku squares
}} 
\newcommand{\koz}{\,}

\begin{document}
\maketitle

\twoauthors{\href{http://compalg.inf.elte.hu/tanszek/index.php}{Antal IV\'ANYI}}
{\href{http://www.elte.hu/en}{E\"otv\"os Lor\'and University} \\ \href{http://compalg.inf.elte.hu/tanszek/index.php?angolul=1}
{Department of Computer Algebra} \\ H-1117, Budapest, Hungary \\ P\'azm\'any s\'etány 1/C}
{\href{mailto:tony@compalg.inf.elte.hu}{tony@compalg.inf.elte.hu}}
{\href{http://compalg.inf.elte.hu/tanszek/index.php}{Imre K\'ATAI}}
{\href{http://www.elte.hu/en}{E\"otv\"os Lor\'and University}\\ \href{http://compalg.inf.elte.hu/tanszek/index.php?angolul=1}
{Department of Computer Algebra} \\ H-1117, Budapest, Hungary \\ P\'azm\'any s\'etány 1/C}  
{\href{mailto:katai@compalg.inf.elte.hu}{katai@compalg.inf.elte.hu}}


\short{%
A. Iv\'anyi, I. K\'atai
}{%
Testing of sequences by simulation}

\bigskip
\begin{abstract}
Let $n$ be a positive integer and $X = [x _{ij}]_{1 \leq i, j \leq n}$ 
be an $n \times n$\linebreak 
\noindent sized matrix of independent random variables having joint uniform distribution 
$$
\hbox{Pr} \{ x _{ij} = k \hbox{ for } 1 \leq k \leq n \} =  \frac{1}{n}  \quad (1 \leq i, j \leq n) \koz.
$$
A realization $\mathcal{M} = [m_{ij}]$ of $X$ is called \textit{good}, if its each row and each column contains a permutation of 
the numbers $1, \ 2, \ \ldots, \ n$. We present and analyse four typical algorithms which decide whether a given realization is good.  
\end{abstract}

\section{Introduction\label{section:1}}
Some subsets of the elements of Latin squares \cite{Adams2008,Buchanan2000,Easton2001,Euler2010,Hajira2007,
Kuhl2008,Kumar1999,Ohman2009,Pedersen2009}, of Sudoku squares \cite{Arcos2010,Bailey2008,
Chang2010,Chen2009,Crook2009,Dahl2009,Denes1991,Erickson2010,Gabor2010,Kanaana2010,Keedwell2010,
Lorch2009,Moon2009,Pedersen2009,Provan2009,Sander2009,Soottile2010,Vaughan2009,Xu2009}, of de Bruijn arrays 
\cite{Alhakim2010,AnisiuBK2002,AnisiuI2006,AnisiuK2011,BondI1987,Brett2009,Cooper2010,Elzinga2008,Elzinga2010,
Horvath2008,Ivanyi1990,Ivanyi2011,IvanyiM2011,IvanyiT1988,Kasa2010,Knuth2011,Matamala2009,
Penne2010,RowleyB1993,Troyanskaya2002,Xu2010,Zhang2009} and gerechte designs, connected with agricultural and 
industrial experiments \cite{Bailey2008,Behrens1956,Heppes1956} have to contain different elements. The one dimensional 
special case is also studied is several papers \cite{FerencziK1999,Hellerman1967,Ivanyi1987,Ivanyi1988,Ivanyi1990,IvanyiK1976,
IvanyiK1978,Kasa1990,Kasa1998,Kasa2011}. 

The testing of these matrices raises the following problem.  

Let $m \geq 1$ and $n \geq 1$ be integers and $X = [x _{ij}]_{1 \leq i \leq m, 1 \leq j \leq n}$ 
be an $m \times n$ sized matrix of independent random variables having joint uniform distribution 
$$
\hbox{Pr} \{ x _{ij} = k \hbox{ for } 1 \leq k \leq n \} =  \frac{1}{n}  
\quad (1 \leq i \leq m, 1 \leq j \leq n)\koz .\label{equation1:1}
$$

A realization $\mathcal{M} = [m_{ij}]$ of $X$ is called \textit{good}, if its each row and each column contain different 
elements (in the case $m = n$ a permutation of the numbers $1, \ 2, \ \ldots, \ n$. We present and analyse algorithms 
which decide whether a given realization is good. If the realization is good then the output of the algorithms is 
\textsc{True}, otherwise is \textsc{False}.

The structure of the paper is as follows. Section \ref{section:1} contains the introduction. 
In Section \ref{section:2} the mathematical background of the main results is prepared. 
Section \ref{section:3} contains the running times of the testing algorithms \textsc{Linear}, \textsc{Backward},   
\textsc{Bucket} and \textsc{Matrix} in worst, best and expected cases. In Section \ref{section:4} the results are summarised.  

\section{Mathematical background}\label{section:2}
We start with the first step of the testing of $\mathcal{M}$: describe and analyse several algorithms testing 
the first row of $\mathcal{M}$. The inputs of these algorithms are $n$ (the length of the first row of $\mathcal{M}$) 
and the elements of the first row  $\mathbf{m} = (m_{11},m_{12},\ldots,m_{1n})$. For the simplicity we use the notation 
\textbf{s} = $(s_1,s_2,\ldots,s_n)$.  
 The output is always a logical variable $g$ (its value is \textsc{True}, 
if the input sequence is good, and \textsc{False} otherwise). 

We will denote the binomial coefficient $\binom{n}{k}$ by $B(n,k)$ and the function $\log _2 n$ by $\lg n$ \cite{CormenLe2009},
and usually omit the argument $n$ from the functions $\tau(n), \ \sigma(n),$  
$\kappa(n), \ \kappa _1(n), \ \kappa _2(n), \ \gamma(n), \ \lambda(n), \ \delta(n), \ \alpha(n), \ \mu(n), 
\eta(n), \ \phi(n), \rho(n), \ \beta(n),$ $S_i(n), \ R_i(n), \ Q(n), p_k(n), \ y(n), \ q_i(k,n), \ A_i(n), \ b_j(n), \ f(n), 
\ p(i,j,k,n)$, $c_j(n), \ c(n)$, and $A(i_1,i_2,k,n)$.

We characterise the running time of the algorithms by the number of necessary 
assignments and comparisons and denote the running time of algorithm \textsc{Alg} 
by $T_{worst}(n,\textsc{Alg}), \ T_{best}(n,\textsc{Alg})$ 
and $T_{exp}(n,\textsc{Alg})$ in the worst, best, resp. expected case. The numbers of the corresponding 
assignments and comparisons are denoted by $A$, resp. $C$. The notations $O, \ \Omega, \ \Theta, \ o$ and $\omega$ 
are used according to \cite[pages 43--52]{CormenLe2009} and \cite[pages 107--110]{Knuth1997}. 

Before the investigation of the concrete algorithms we formulate several lemmas.  
The first lemma is the following version of the well-known Stirling's formula.
\begin{lemma} [\cite{CormenLe2009}] If $n \geq 1$ then\label{lemma:1}
\begin{equation}
n! = \left ( \frac{n}{e} \right )^n \sqrt{2 \pi n} e^{\tau} \koz, \label{equation2:2}
\end{equation}
where
$$
\frac{1}{12n+1} < \tau < \frac{1}{12n} \koz, \label{equation2:3}
$$
and $\tau(n) = \tau$ tends monotonically decreasing to zero when $n$ tends to infinity.
\end{lemma}

Let $a_k(n) = a_k$ and $S_i(n) = S_i$ defined for any positive integer $n$ as follows:
$$
a_k = \frac{n^k}{k!} \quad (k = 0, \ 1, \ 2, \ \ldots) \koz, \label{equation2:4}
$$

\begin{equation}
S_i = \sum _{k=0}^{n-1} a_k k^i \quad (i = 0, \ 1, \ 2, \ \ldots) \koz. \label{equation2:5}
\end{equation}

If in (\ref{equation2:5}) $k = i =0$, then $k^i =0$. 

Solving a problem posed by S. Ramanujan \cite{Ramanujan1928}, G\'abor Szeg\H o \cite{Szego1928} proved 
the following connection between $e^n$ and $S_0$.
\begin{lemma} [\cite{Szego1928}] The \label{lemma:2} function $\sigma(n) = \sigma,$ defined by 
\begin{equation}
\frac{e^n}{2} = S_0 +  \left ( \frac{1}{3} + \sigma \right ) a_n = \sum _{k=0}^{n-1} \frac{n^k}{k!} + 
\left ( \frac{1}{3} + \sigma \right ) a_n \quad (n = 1, \ 2, \ \ldots) \label{equation2:6}
\end{equation}
and
$$
\sigma(0) = \frac{1}{6} \koz, \label{equation2:7}
$$
tends monotonically decreasing to zero when $n$ tends to $\infty$.
\end{lemma}

The following lemma shows the connection among $S_i$ and $S_0, \ S_1, \ \ldots, \ S_{i-1}$. 

\begin{lemma} If\label{lemma:3} $i$ and $n$ are positive integers, then
\begin{equation}
S_i = n\sum _{k=0}^{i-1} B(i-1,k)S_k - n^{i-1}a_{n-1} \label{equation2:8}
\end{equation}
and
\begin{equation}
S_i = \Theta(e^n n^i) \koz. \label{equation2:9}
\end{equation}
\end{lemma}
\begin{proof}
Omitting the member belonging to the index $k = 0$ in $S_i$, then simplifying by $k$ and using the substitution 
$k - 1 = j$ we get
$$
S_i = \sum _{k=0}^{n-1} \frac{n^k}{k!} k^i = n \sum _{k=1}^{n-1} \frac{n^{k - 1}}{(k - 1)!} k^{i-1} = 
n \sum _{j=0}^{n-2} \frac{n^j}{j!} (j+1)^{i-1} \koz. \label{equation2:10} 
$$
Completing the sum with the member belonging to index $j = n - 1$ results
\begin{equation}
S_i = n \sum _{j=0}^{n-1} \frac{n^j}{j!} (j+1)^{i-1} - n^ia_{n-1} \koz. \label{equation2:11} 
\end{equation}
Now the application of the binomial theorem results (\ref{equation2:8}). 

According to (\ref{equation2:9}) $S_0 = \Theta(e^n)$, so using induction and (\ref{equation2:11}) we get (\ref{equation2:9}).
\end{proof}

In this paper we need only the simple form of $S_0, \ S_1, \ S_2$ and $S_3$ what is presented in the next lemma.
\begin{lemma} If\label{lemma:4} $n$ is a positive integer then
\begin{equation}
S_0 = \frac{e^n}{2} - \frac{n^n}{n!}\left ( \frac{1}{3} + \sigma \right ) \koz,\label{equation2:12}
\end{equation}
\begin{equation}
S_1 = nS_0 - n a_{n-1}, \quad S_2 = S_0(n^2 + n) - 2n^2a_n \koz, \label{equation2:13}
\end{equation}
and
\begin{equation}
S_3 = S_0(n^3 + 3n^2 + n) - (3n^3 + 2n^2)a_n \koz. \label{equation2:14}
\end{equation}
\end{lemma} 
\begin{proof}
Expressing $S_0$ from (\ref{equation2:6}), and using recursively Lemma \ref{lemma:3} for $i =1, \ 2$ and $3$ we get 
the required formula for $S_0, \ S_1, \ S_2,$ and $S_3.$
\end{proof}

We introduce also another useful function $R_i(n) = R_i$
\begin{equation}
R_i = \sum _{k=1}^{n} p_k(n) k^i \quad (i = 0, \ 1, \ 2, \ \ldots) \koz,\label{equation2:15}
\end{equation}
where $p_k(n) = p_k$ is the key probability of this paper, defined in \cite{Hellerman1967} as
\begin{equation}
p_k = \frac{n}{n} \frac{n-1}{n} \cdots \frac{n-k+1}{n} \frac{k}{n} =  
\frac{n!k}{(n - k)!n^{k+1}} \quad (k = 1, \ 2, \ \ldots, \ n) \koz. \label{equation2:16}
\end{equation}

The following lemma mirrors the connection between the function $R_i$ and the functions $S_0, \ S_1, \ \ldots, \ S_{i+1}.$

\begin{lemma} If\label{lemma:5} $i$ and $n$ are positive integers, then 
\begin{equation}
R_i = \frac{n!}{n^{n+1}} \sum _{l=0}^{i+1} (-1)^l \binom{i+1}{l} n^{i+1-l} S_l \koz. \label{equation2:17}
\end{equation}
\end{lemma}

\begin{proof}
Using (\ref{equation2:15}) and (\ref{equation2:16}) the substitution $n - k = j$ results
$$
R_i = \sum _{k=1}^n \frac{n!k^{i+1}}{(n - k)!n^{k+1}} = 
\frac{n!}{n^{n+1}} \sum _{j=0}^{n-1} \frac{n^j(n - j)^{i+1}}{j!} \koz .\label{equation2:18}
$$

From here, using the binomial theorem we get (\ref{equation2:17}).
\end{proof}  

In this paper we need only the following consequence of Lemma \ref{lemma:5}.

\begin{lemma} If\label{lemma:6} $n$ is a positive integer, then
$$
R_0 = 1, \quad  R_1 = \frac{n!}{n^n} S_0 \koz,\label{equation2:19}
$$
and
\begin{equation}
R_2 = 2n - \frac{n!}{n^n}S_0 \koz. \label{equation2:20}  
\end{equation}
\end{lemma}
\begin{proof}
$R_0 = 0$ follows from the definition of the probabilities $p_k.$ Substituting $i = 1$ into (\ref{equation2:17}) we get
$$
R_1 = \frac{n!}{n^{n+1}} \left (n^2 \sum _{j=0}^{n-1} \frac{n^j}{j!} -2n \sum _{j=0}^{n-1} \frac{n^j}{j!}j 
+ \sum _{j=0}^{n-1} \frac{n^j}{j!}j^2 \right ) \koz.\label{equation2:21}
$$
From here, using (\ref{equation2:5}) we get 
$$
R_1 = \frac{n!}{n^{n+1}} (n^2S_0 - 2nS_1 + S_2) \koz,\label{equation2:22}
$$
and using (\ref{equation2:11}) the required formula for $R_1$.

Substituting $i = 2$ into (\ref{equation2:17}) we get
$$
R_2 = \frac{n!}{n^{n+1}} \left (n^3 \sum _{j=0}^{n-1} \frac{n^j}{j!} -3n^2 \sum _{j=0}^{n-1} \frac{n^j}{j!}j 
+ 3n \sum _{j=0}^{n-1} \frac{n^j}{j!}j^2 - \sum _{j=0}^{n-1} \frac{n^j}{j!}j^3 \right ) \koz.\label{equation2:23}
$$

From here, using (\ref{equation2:5}) we have
\begin{equation}
R_2 = \frac{n!}{n^{n+1}} (n^3S_0 - 3n^2S_1 + 3nS_2 - S_3) \koz,\label{equation2:24}
\end{equation}
and using (\ref{equation2:13}) and (\ref{equation2:14}) the required formula for $R_2$.
\end{proof}

The following lemmas give some further properties of $R_1$ and $R_2$.

\begin{lemma} If\label{lemma:7} $n$ is a positive integer, then
\begin{equation}
R_1 = \frac{n!}{n^n} S_0 = \sqrt{\frac{\pi n}{2}} - \frac{1}{3} + \kappa \koz, \label{equation2:25}  
\end{equation}
where
\begin{equation}
\kappa(n) = \kappa = \sqrt{\frac{\pi n}{2}} \left( e^{\tau} - 1 - \frac{2\sigma e^{\tau}}{e^n} \right) \koz, \label{equation2:26}
\end{equation}
and $\kappa$ tends monotonically decreasing to zero when $n$ tends to infinity.
\end{lemma}
\begin{proof}
Substituting $S_0$ according to (\ref{equation2:12}) in the formula (\ref{equation2:20}) for $R_1$  we get
\begin{equation}
R_1 = \frac{n!}{n^n}\left[ \frac{e^n}{2} - \frac{n^n}{n!}\left ( \frac{1}{3} + \sigma \right) \right ] = 
- \frac{1}{3} +  \frac{n!}{n^n}\left( \frac{e^n}{2} - \frac{n^n}{n!} \sigma \right) \koz. \label{equation2:27}
\end{equation}

Substitution of $n!$ according to (\ref{equation2:2}) (Stirling's formula) and writing $1 + (e^{\tau}-1)$ 
instead of $e^{\tau}$ results
\begin{equation}
R_1 = - \frac{1}{3} + \frac{1}{n^n} \left ( \frac{n}{e} \right )^n \sqrt{2 \pi n} 
\left[ 1 + (e^{\tau} - 1) \right ] \left[\frac{e^n}{2} - \sigma \right ] \koz. \label{equation2:28}
\end{equation}

The product $P$ of the expressions in the square brackets is 
\begin{equation} 
P = \frac{e^n}{2} + \frac{e^n}{2} \left( e^{\tau} - 1 \right) - \sigma e^{\tau} \koz,  \label{equation2:29}
\end{equation}
therefore
\begin{equation}
R_1 = \sqrt{\frac{\pi n}{2}} - \frac{1}{3} + \frac{\sqrt{2\pi n}}{e^n} \left [\frac{e^n}{2}\left ( e^{\tau} - 1 \right ) 
- \sigma e^{\tau} \right] \koz, \label{equation2:30}
\end{equation}
implying
\begin{equation}
R_1 =  \sqrt{\frac{\pi n}{2}} - \frac{1}{3} + \sqrt{\frac{\pi n}{2}}\left(e^{\tau} - 1 \right) 
- \sqrt{\frac{\pi n}{2}}\frac{2\sigma e^{\tau}}{e^n} \koz . \label{equation2:31}
\end{equation}

Let
\begin{equation}
\kappa _1(n) = \kappa _1 = \sqrt{\frac{\pi n}{2}}(e^{\tau} - 1), \  \kappa _2(n) = \kappa _2 =  
\sqrt{\frac{\pi n}{2}}\frac{2\sigma e^{\tau}}{e^n},  \ \kappa = \kappa _1 + \kappa _2 \koz,\label{equation2:32}
\end{equation}
and
\begin{equation}
\gamma(n) = \gamma = \frac{\kappa(n + 1)}{\kappa(n)} = \frac{\kappa _1(n + 1) - \kappa _2(n + 1)}
{\kappa _1(n) - \kappa _2(n)} \quad \hbox{for } n = 1, \ 2, \ \ldots \koz . \label{equation2:33}
\end{equation}

Since all $\kappa$ functions are positive for all positive integer  $n$'s, therefore $\gamma < 1$ for $n \geq 1$ 
implies the monotonity of $\kappa$. Numerical results in Table \ref{table:1} show that $\gamma < 1$ 
for $n = 1, \ 2, \ \ldots, \ 9$, therefore it remained to show $\gamma < 1$ for $n \geq 10$.

$\kappa _2(n + 1)$ can be omitted from the numerator of (\ref{equation2:32}). Since $\sigma$ and $\tau$ are monotone 
decreasing functions, and $0 < \sigma(5) < 0.0058$, and $0 < e^{\tau(5)} < 1.02,$ and $n^2 < e^n$ for $n \geq 10,$ therefore
\begin{equation}
\frac{2\sigma e^{\tau}}{e^n} < \frac{2\cdot 0.0058\cdot 1.02}{e^n} < \frac{0.012}{n^2} 
\hbox{ for } n \geq 10 \koz. \label{equation2:34}
\end{equation}

Using (\ref{equation2:33}), (\ref{equation2:34}) and the Lagrange remainder of the Taylor series of the function $e^x$ we have
$$
\gamma < \frac{\sqrt{n+1}}{\sqrt{n}} \frac{\tau(n+1) + \tau^2 \xi _{n+1}/2}{\tau(n) + \tau^2 \xi _n/2   
- 0.012/n^2} \koz , \label{equation2:35}
$$
where $0 < \xi _{n+1} < n +1$ and $0 < \xi _{n} < n$, therefore using Lemma \ref{lemma:1} we get
\begin{equation}
\gamma < \frac{\sqrt{n+1}}{\sqrt{n}} \frac{\frac{1}{12(n + 1) + 1}}{\frac{1}{12n} 
+ \frac{1}{2}\left ( \frac{1}{12n} \right )^2 - \frac{0.012}{n^2}} \koz . \label{equation2:36}
\end{equation}

Now multiplication of the denominator and denominator of the right side of (\ref{equation2:36}) by $(12n)^2$ results
\begin{equation}  
\gamma =  \frac{\sqrt{n+1}}{\sqrt{n}} \frac{\frac{12n\cdot 12n}{12n+13}}{12n + 0.5 - 1.584}  
= \frac{\sqrt{n+1}}{\sqrt{n}} \frac{12n}
{\left (12n - 1.084 \right ) \left(1 + \frac{13}{12n} \right)} \koz . \label{equation2:37}
\end{equation}

Since
\begin{equation}
\left (12n - 1.084 \right ) \left(1 + \frac{13}{12n} \right) > 12n + 10 \koz, \label{equation2:38}
\end{equation}
(\ref{equation2:37}) and (\ref{equation2:38}) imply 
$$
\gamma < \frac{\sqrt{144n^3 + 144n^2}}{\sqrt{144n^3 + 240 n^2}} < 1 \koz, \label{equation2:39}
$$
finishing the proof of the monotonity of $\kappa$.
\end{proof}
We remark, that the monotonity of $\kappa$ was published in \cite{IvanyiK1976} without proof, and was proved by E. Bokova and G. Tzaturjan in 1985 \cite{BokovaT1985}, and in 1988---using a formula due to E. Egorychev et al. \cite{EgoricheIM1987} derived by the method of integral representation of combinatorial sums elaborated by E. P. Egorychev \cite{Egorychev1984}---by T. T. Cirulis and A. Iv\'anyi \cite{Cirulis1987}. Our proof is much simpler than the earlier ones.

\begin{lemma} If\label{lemma:8} $n$ is a positive integer, then
$$
R_2 = 2n - \frac{n!}{n^n} S_0 = 2n + \frac{1}{3} - \sqrt{\frac{\pi n}{2}}e^{\tau} - \lambda \koz, \label{equation2:40}  
$$

\begin{equation}
\hspace*{-2.2cm}\textrm{\it where} \qquad\quad\qquad \lambda = \sqrt{\frac{\pi n}{2}}(e^{\tau} - 1) + \sigma \koz, \qquad \qquad \qquad \qquad \qquad  \label{equation2:41}
\end{equation}
and $\lambda$ tends monotonically decreasing to zero when $n$ tends to infinity.
\end{lemma}
\begin{proof} The proof is omitted since it is similar to the proof of Lemma \ref{lemma:7}.
\end{proof}

\section{Running times of the algorithms\label{section:3}}
In the following analysis let $n \geq 1$ and let \textbf{x} = $(x _1, \ x _2, \ \ldots, \ x _n)$ be independent random 
variables having uniform distribution on the set $\{1,2, \ldots , n \}$. The input sequence of the algorithms is 
$\mathbf{s} = (s_1, s_2, \ldots, s_n)$ (a realization of \textbf{x}).

We derive exact formulas for the expected numbers of comparisons $C_{exp}(n,$\linebreak 
\noindent $\textsc{Linear})$ $ = C_L$, 
$C_{exp}(n,\textsc{Backward}) = C_W$, and $C_{exp}(n,\textsc{Bucket}) = C_B$, further 
for the expected running times $T_{exp}(n,\textsc{Linear}) = T_L$, 
$T_{exp}(n,$ \linebreak 
\noindent $\textsc{Backward}) = T_W$, and $T_{exp}(n,\textsc{Bucket}) = T_B$. 

The inputs of the following algorithms are $n$ (the length of the sequence \textbf{s}) and 
$\mathbf{s} = (s_1, s_2, \ldots, s_n),$ a sequence of nonnegative integers with $1 \leq s_i \leq n$ for 
$1 \leq i \leq n$) in all cases. The output is always a logical variable $g$ (its value is \textsc{True}, 
if the input sequence is good, and \textsc{False} otherwise). 
The working variables are usually the cycle variables $i$ and $j$. 

We use the pseudocode defined in \cite{CormenLe2009}.

\subsection{Definition and running time of algorithm \sc{Linear}\label{subsection:31}}
\textsc{Linear} writes zero into the elements of an $n$ length vector $\mathbf{v} = (v_1, \ v_2,$  
$ \ \ldots, \ v_n)$, then investigates the elements of the realization \textbf{s} and if $v_{s_i} > 0$ (signalising a repetition), 
then returns \textsc{False}, otherwise adds 1 to $v_k$. If \textsc{Linear} does not find a repetition among the elements 
of \textbf{s} then it returns finally \textsc{True}. 

\medskip
\noindent \textsc{Linear}$(n,\mathbf{s})$
\vspace{-2mm}
\begin{tabbing}%
199 \= xxx\=xxx\=xxx\=xxx\=xxx\=xxx\=xxx\=xxx \+ \kill 
\hspace{-7mm}\textbf{1} $g \leftarrow$ \textsc{True} \\
\hspace{-7mm}\textbf{2} \textbf{for} \= $i \leftarrow 1$ \textbf{to} $n$ \\
\hspace{-7mm}\textbf{3}              \> $v_i \leftarrow 0$ \\
\hspace{-7mm}\textbf{4} \textbf{for} \= $i \leftarrow 1$ \textbf{to} $n$ \\             
\hspace{-7mm}\textbf{5}              \> \textbf{if}   \=$v_{s_i} > 0$ \\
\hspace{-7mm}\textbf{6}              \>             \> $g \leftarrow$ \textsc{False} \\
\hspace{-7mm}\textbf{7}              \>             \> \textbf{return} $g$ \\ 
\hspace{-7mm}\textbf{8}              \> \textbf{else}  $v_{s_i} \leftarrow v_{s_i} + 1$ \\     
\hspace{-7mm}\textbf{9} \textbf{return} $g$
\end{tabbing}

\textsc{Linear} needs assignments in lines 1, 3, and 8, and it needs comparisons in line 5. The number of assignments 
in lines 1 and 3 equals to $n + 1$ for arbitrary input and varies between $1$ and $n$ in line 8. 
The number of comparisons in line 8 also varies between $1$ and $n$. Therefore the running time of \textsc{Linear}  
is $\Theta(n)$ in the best, worst and expected case too. 

The following theorem gives the expected number of the comparisons of \textsc{Linear}.

\begin{theorem} The\label{theorem:9} expected number of comparisons $C_{exp}(n,\textsc{Linear}) = C_L$ of \textsc{Linear} is
\begin{equation}
C_L = 1 - \frac{n!}{n^n} + R_1 = \sqrt{\frac{\pi n}{2}} + \frac{2}{3} + \kappa - \frac{n!}{n^n} \koz . \label{equation31:42}
\end{equation}
where 
$$
\kappa = \frac{1}{3} - \sqrt{\frac{\pi n}{2}} + \sum _{k=1}^n \frac{n!k^2}{(n - k)!n^{k+1}}
$$ 
tends monotonically decreasing to zero when $n$ tends to infinity.
\end{theorem}
\begin{proof}
Let 
\begin{equation}
y(n) = y = \max \{k: 1 \leq k \leq n \hbox{ and } s_1, \ s_2, \ \ldots, \ s_k \hbox{ are different} \} \label{equation31:43}
\end{equation}
be a random variable characterising the maximal length of the prefix of $\mathbf{s}$ containing different elements. Then 
$$
\hbox{Pr} \{ y = k \} = p_k \quad (k = 1, \ 2, \ \ldots, \ n) \koz, \label{equation31:44}
$$
where $p_k$ is the probability introduced in (\ref{equation2:16}). 
 
If $y = k$ and $1 \leq k \leq n - 1,$ then \textsc{Linear} executes $k+1$ comparisons, and only $n$ comparisons, if 
$y = n,$ therefore 
\begin{equation}
C_L = \sum _{k=1}^{n-1} p_k(k + 1) + p_n n =  \sum _{k=1}^n p_k(k + 1) - p_n = 1 
- \frac{n!}{n^n} + \sum_{k=1}^n p_k k \koz, \label{equation31:45} 
\end{equation}    
from where using Lemma \ref{lemma:7} we receive 
\begin{equation}
C_L = 1 - \frac{n!}{n^n} + R_1 = \sqrt{\frac{\pi n}{2}} + \frac{2}{3} - \frac{n!}{n^n} + \kappa \koz . \label{equation31:46}
\end{equation}

The monotonity of $\kappa(n)$ was proved in the proof of Lemma \ref{lemma:7}.
\end{proof}

The next assertion gives the running time of \textsc{Linear}.

\begin{theorem} The\label{theorem:10} expected running time $T_{exp}(n,\textsc{Linear}) = T_L$ of \textsc{Linear} is
$$
T_L = n + \sqrt{2\pi n} + \frac{7}{3} + 2 \kappa - 2 \frac{n!}{n^n} \koz, \label{equation31:47}
$$
where $\kappa$ tends monotonically decreasing to zero when $n$ tends to infinity.
\end{theorem}
\begin{proof}
\textsc{Linear} requires $n + 1$ assignments in lines 01 and 03, plus assignments in line 08. The expected number 
of assignments in line 8 is the same as $C_L$. Therefore
\begin{equation}
T_L = n + 1 + 2 C_L \koz.\label{equation31:48}
\end{equation}

Substitution of (\ref{equation31:46}) into (\ref{equation31:48}) results the required (\ref{equation31:42}).
\end{proof}

We remark, that (\ref{equation31:46}) is equivalent with 
$$
C_L = 1 - \frac{n!}{n^n} + 1 + \frac{n-1}{n}+\frac{n-1}{n}\frac{n-2}{n} + 
\cdots + \frac{n-1}{n}\frac{n-2}{n} \cdots \frac{1}{n} \koz, \label{equation31:49}
$$
demonstrating the close connection with the function
\begin{equation}
Q(n) = Q = C_L - 1 + \frac{n!}{n^n} \koz, \label{equation31:50}
\end{equation}
studied by several authors, e.g. in \cite{Breusch1968,IvanyiK1976,Knuth1997}. 

Table \ref{table:1} shows the concrete values of the functions  appearing in the analysis of $C_L$ and 
$T_L$ for $1\leq n\leq 10$, where $C_L$ was calculated using (\ref{equation31:46}), $\kappa$ using 
(\ref{equation2:16}), and $\sigma$ using (\ref{equation2:6}) (data in this and further tables are taken 
from \cite{IvanyiN2010}). We can observe in Table \ref{table:1} that 
$\delta(n) = \delta = \kappa - \frac{n!}{n^n}$ 
is increasing from $n = 1$ to $n = 8,$ but for larger $n$ is decreasing. Taking into account that for $n > 8$  

\begin{table}[!t] 
\begin{center}
\begin{tabular}{|c|c|c|c|c|c|c|c|}    \hline
$n$   &     $C_L$   &    $u$    & $n!/n^n$  &$\kappa$& $\delta$ &$\sigma$  \\ \hline
$1$   &$1.000000$ &    $1.919981$                & $1.000000$&$0.080019$ & $-0.919981$                      &0.025808    \\ \hline
$2$   &$2.000000$ &    $2.439121$                & $0.500000$&$0.060879$ & $-0.439121$                      &0.013931    \\ \hline
$3$   &$2.666667$ &    $2.837470$                & $0.222222$&$0.051418$ & $-0.170804$                      &0.009504    \\ \hline
$4$   &$3.125000$ &    $3.173295$                & $0.093750$&$0.045455$ & $-0.048295$                      &0.007205    \\ \hline
$5$   &$3.472000$ &    $3.469162$                & $0.038400$&$0.041238$ & $+0.002838$                      &0.005799    \\ \hline
$6$   &$3.759259$ &    $3.736647$                & $0.015432$&$0.038045$ & $+0.022612$                      &0.004852    \\ \hline
$7$   &$4.012019$ &    $3.982624$                & $0.006120$&$0.035515$ & $+0.029395$                      &0.004170    \\ \hline
$8$   &$4.242615$ &    $4.211574$                & $0.002403$&$0.033444$ & $+0.031040$                      &0.003656    \\ \hline
$9$   &$4.457379$ &    $4.426609$                & $0.000937$&$0.031707$ & $+0.030770$                      &0.003255    \\ \hline
$10$  &$4.659853$ &    $4.629994$                & $0.000363$&$0.030222$ & $+0.029859$                      &0.002933    \\ \hline
\end{tabular}
\caption{Values of $C_L,$ $u = \sqrt{\pi n/2} + 2/3$, $n!/n^n,$ $\kappa,$ $\delta = \kappa - n!/n^n,$ and $\sigma$ 
for $n = 1, \ 2, \ \ldots, \ 10$ \label{table:1}}
\end{center}
\end{table}
$$
\frac{n!}{n^n} = \left (\frac{n}{e} \right )^n \sqrt{2\pi n} \frac{e^{\tau}}{n^n} < \frac{\sqrt{2 \pi n}}{e^n} e^{1/(12n)}  
< \frac{2.7\sqrt{n}}{e^n} < \frac{0.012}{n^2} \label{equation31:51}
$$
holds, we can prove---using the same arguments as in the proof of Lemma \ref{lemma:7}---the following assertion.
\begin{theorem} The\label{theorem:11} expected running time $T_{exp}(n,\textsc{Linear}) = T_L$ of \textsc{Linear} is
$$
T_L = n + \sqrt{2\pi n} + \frac{7}{3} + \delta \koz, \label{equation31:52}
$$
where $\delta(n) = \delta$ tends to zero when $n$ tends to infinity, further 
$$
\delta(n + 1) > \delta(n) \mbox{ for } 1 \leq n \leq 7 \mbox{ and } \delta(n + 1) 
< \delta(n) \hbox{ for } n \geq 8 \koz. \label{equation31:53}
$$
\end{theorem}

If we wish to prove only the existence of some threshold index $n_0$ having the property that $n \geq n_0$ implies 
$\delta(n+1) < \delta(n)$, then we can use the following shorter proof.

Using (\ref{equation31:42}) and (\ref{equation31:50}) we get
\begin{equation}
\kappa = C_L - \sqrt{\frac{\pi n}{2}} - \frac{2}{3} - \frac{n!}{n^n} = Q - \sqrt{\frac{\pi n}{2}} 
+ \frac{1}{3} \koz.\label{equation31:54}
\end{equation}

Substituting the power series 
$$
Q = \sqrt{\frac{\pi n}{2}} - \frac{1}{3} + \frac{1}{12}\frac{\pi}{2n} - \frac{14}{135n} + \frac{1}{288}\frac{\pi}{2n^3} 
+ O(n^{-2}) \label{equation31:55}
$$
cited by D. E. Knuth \cite[Equation (25) on page 120]{Knuth1997} into (\ref{equation31:54}) and using  
$$
\frac{1}{n^{k/2}} - \frac{1}{(n+1)^{k/2}} = \Theta \left ( \frac{1}{n^{1+k/2}} \right) \label{equation31:56}
$$
for $k = 1, \ 2, \ 3$ and $4$ we get
$$
\kappa(n) - \kappa(n+1) = \frac{\sqrt{\pi}}{12\sqrt{2}} \left( \frac{1}{\sqrt{n}} 
- \frac{1}{\sqrt{n+1}} \right) + O(n^{-2}) \koz, \label{equation31:57}
$$
implying
$$ 
\kappa(n) - \kappa(n+1) = \frac{\sqrt{\pi}}{12\sqrt{2}} \frac{1}{\sqrt{n}\sqrt{n+1}(\sqrt{n}+\sqrt{n+1})}
 + O(n^{-2}) \koz,\label{equation31:58}
$$
guaranteeing the existence of the required $n_0$.
 
\subsection{Running time of algorithm \sc{Backward} \label{subsection:32}}
\textsc{Backward} compares the second $(s_2)$, third $(s_3)$, \ldots, last $(s_n)$ element of the realization with 
the previous elements until the first collision or until the last pair of elements. 

Taking into account the number of the necessary comparisons in line 04 of \textsc{Backward}, we get   
$C_{best}(n,\textsc{Backward}) = 1 = \Theta(1)$, and $C_{worst}(n,$ \linebreak 
$\textsc{Backward}) = B(n,2) = \Theta(n^2)$. The number of 
assignments is 1 in the best case (in line 1) and is 2 in the worst case (in lines 1 and in line 5). The expected number 
of assignments is $A_{exp}(n,\textsc{Backward}) = 1 + \frac{n!}{n^n}$, since only the good realizations 
require the second assignment.

\medskip
\noindent \textsc{Backward}$(n,\mathbf{s})$
\vspace{-2mm}
\begin{tabbing}%
199 \= xxx\=xxx\=xxx\=xxx\=xxx\=xxx\=xxx\=xxx \+ \kill 
\hspace{-7mm}\textbf{1} $g \leftarrow$ \textsc{True} \\
\hspace{-7mm}\textbf{2} \textbf{for} \= $i \leftarrow 2$ \textbf{to} $n$ \\
\hspace{-7mm}\textbf{3}              \> \textbf{for} \= $j \leftarrow i - 1$ \textbf{downto} $1$ \\
\hspace{-7mm}\textbf{4}              \>              \> \textbf{if} \= $s_i = s_j$ \\
\hspace{-7mm}\textbf{5}              \>              \>             \> $g \leftarrow$ \textsc{False} \\
\hspace{-7mm}\textbf{6}              \>              \>             \> \textbf{return} $g$ \\ 
\hspace{-7mm}\textbf{7} \textbf{return} $g$
\end{tabbing}

The next assertion gives the expected running time.

\begin{theorem} The expected\label{theorem:12} number of comparisons $C_{exp}(n,\textsc{Backward}) = C_W$ 
of the algorithm \textsc{Backward} is 
$$
C_W = n - \sqrt{\frac{\pi n}{8}} + \frac{2}{3} - \frac{1}{2}\kappa - \frac{n!}{n^n} \frac{n + 1}{2} 
= \sqrt{\frac{\pi n}{8}} + \frac{2}{3} - \alpha \koz,\label{equation32:59}
$$
where $\alpha(n) = \alpha = \frac{\kappa}{2} + \frac{n!}{n^n}\frac{n+1}{2}$ monotonically decreasing tends 
to zero when $n$ tends to $\infty$.
\end{theorem}
\begin{proof}
Let $y$ be as defined in (\ref{equation31:43}), $p_k$ as defined in (\ref{equation2:16}), and let 
$$
z = \{q : 1 \leq q \leq k; \ s_1,s_2,\ldots,s_k \hbox{ are different; } 
s_{k+1} = s_q \ | \ y = k \} \label{equation32:60}
$$ 
be a random variable characterising the index of the first repeated element of \textbf{s}. 

Let
$$
q_i(k,n) = q_i(k) = \hbox{Pr} \{z = i | y = k\} \quad (k = 1, \ 2, \ \ldots, \ n; \ i = 1, \ 2,
 \ \ldots k) \koz. \label{equation32:61}
$$
 
\textsc{Backward} executes $B(k,2)$ comparisons among the elements $s_1, \ s_2, \ \ldots, \ s_k,$ and $s_{k+1}$ 
requires at least 1 and at most $k$ comparisons (with exception of case $k = n$ when additional comparisons are 
not necessary). Therefore using the theorem of the full probability we have
$$
C_W = \sum_{k=1}^{n-1} p_k \left (B(k,2) + \sum_{i=1}^k i q_i(k) \right ) + p_n B(n,2) \koz,\label{equation32:62}
$$
where
\begin{equation}
q_i(k,n) = q_i(k) = \frac{1}{k}  \quad (i = 1, \ 2, \ \ldots, \ k; \ k = 1, \ 2, \ldots, \ n) \koz. \label{equation32:63}
\end{equation}
Adding a new member to the first sum we get
\begin{equation} 
C_W = \sum_{k=1}^n p_k \left (B(k,2) + \sum_{i=1}^k q_i(k) i \right ) - p_n\sum _{i=1}^{n} q_i(k) i \koz. \label{equation32:64} 
\end{equation}
Using the uniform distribution (\ref{equation32:63}) of $z$ we can determine its contribution to $C_W$:
\begin{equation}
\sum_{i=1}^k  q_i(k) i = \sum_{i=1}^k  \frac{i}{k} = \frac{k + 1}{2} \koz.\label{equation32:65}
\end{equation}
Substituting the contribution in  (\ref{equation32:65}) into (\ref{equation32:64}), and taking into account Lemma 
\ref{lemma:6} and Lemma \ref{lemma:7} we have
$$
C_W = \frac{1}{2}R_2 - \frac{1}{2} R_0 - \frac{n!}{n^n} \frac{n+1}{2} \koz. \label{equation32:66} 
$$

Now Lemma \ref{lemma:6} and Lemma \ref{lemma:7} result
\begin{equation}
C_W = n - \sqrt{\frac{\pi n}{8}} + \frac{2}{3} - \frac{1}{2}\kappa - \frac{n!}{n^n} \frac{n + 1}{2} \koz.\label{equation32:67}
\end{equation}

The known decreasing monotonity of $\kappa$ and $\frac{n!}{n^n}$ imply the decreasing monotonity of $\alpha$.
\end{proof} 

\begin{theorem} The expected\label{theorem:13} running time $T_{exp}(n,\textsc{Backward}) = T_W$ of the algorithm 
\textsc{Backward} is 
\begin{equation}
T_W = n - \sqrt{\frac{\pi n}{8}} + \frac{5}{3} - \alpha \koz,\label{equation32:68}
\end{equation}
where $\alpha = \kappa/2 + (n!/n^n)((n+1)/2)$ 
tends monotonically decreasing to zero when $n$ tends to $\infty$.
\end{theorem}

\begin{proof} 
Taking into account (\ref{equation32:67}) and $A_{exp}(n,\textsc{Backward}) = 1 + \frac{n!}{n^n} - \frac{n!}{n^n} \frac{n+1}{2}$ 
we get (\ref{equation32:68}). 
\end{proof}

Table \ref{table:2} represents some concrete numerical results. 
It is worth to remark that $\frac{n!}{n^n}\frac{n+1}{2} = \Theta \left( \frac{n\sqrt{n}}{e^n} \right),$  
while $\kappa = \Theta \left( \frac{1}{\sqrt{n}} \right),$ therefore $\kappa$ decreases much slower than the other expression.

\begin{table}[!t] 
\begin{center}
\begin{tabular}{|c|c|c|c|c|c|c|}    \hline
$n$&   $C_W$   &$n - \sqrt{\pi n/8}+2/3$&$t$&$\kappa$&$\alpha$   \\ \hline
$1$&$0.000000$&    $1.040010$         & $1.000000$                  & $0.080019$&$1.040010$     \\ \hline
$2$&$1.000000$&    $1.780440$         & $0.750000$                  & $0.060879$&$0.780440$     \\ \hline
$3$&$2.111111$&    $2.581265$         & $0.444444$                  & $0.051418$&$0.470154$     \\ \hline
$4$&$3.156250$&    $3.413353$         & $0.234375$                  & $0.045455$&$0.257103$    \\ \hline
$5$&$4.129600$&    $4.265419$         & $0.115200$                  & $0.041238$&$0.135819$     \\ \hline
$6$&$5.058642$&    $5.131677$         & $0.054012$                  & $0.038045$&$0,073035$   \\ \hline
$7$&$5.966451$&    $6.008688$         & $0.024480$                  & $0.035515$&$0.042237$    \\ \hline
$8$&$6.866676$&    $6.894213$         & $0.010815$                  & $0.033444$&$0.027536$     \\ \hline
$9$&$7.766159$&    $7.786695$         & $0.004683$                  & $0.031707$&$0.020537$   \\ \hline
$10$&$8.667896$&   $8.685003$         & $0.001996$                  & $0.030222$&$0.017107$   \\ \hline
\end{tabular}
\caption{Values of $C_W,$ $n - \sqrt{\pi n/8}+2/3$, $t = \frac{n!}{n^n}\frac{n+1}{2},$ $\kappa,$ and 
$\alpha = \kappa/2 + (n!/n^n)((n+1)/2)$ for $n = 1, \ 2, \ \ldots, \ 10$ \label{table:2}}
\end{center}
\end{table}

\subsection{Running time of algorithm \sc{Bucket} \label{subsection:33}}
\textsc{Bucket} divides the interval $[1,n]$ into $m = \sqrt{n}$ subintervals 
$I_1, \ I_2, \ldots, \ I_m,$ where $I_j = [(j - 1)m + 1, jm]$ for 
$j = 1, \ 2, \ \ldots \ m$, and sequentially puts the elements of \textbf{s} into the 
bucket $B_j$ (we use the word bucket due to some similarity to bucket sort \cite{CormenLe2009}): if $\lceil s_i/m \rceil = j,$ 
then $s_i$ belongs to $B_j$. \textsc{Bucket} works until the first repetition (stopping with 
$g = \textsc{False}$), or up to the processing of the last element $s_n$ (stopping with $g = \textsc{True}$).

\textsc{Bucket} handles an array $Q[1:m,1:m]$ (where $m = \lceil \sqrt{n} \rceil$  
and puts the element $s_i$ into the $r$th row of $Q$, and  
it tests using linear search whether $s_j$ appeared earlier in the corresponding bucket.  
The elements of the vector $\mathbf{c} = (c_1,c_2,\ldots,c_{m})$ are counters, where $c_j$ $(1 \leq j \leq m))$ 
shows the actual number of elements in $B_j.$

\medskip
\noindent \textsc{Bucket}$(n,\mathbf{s})$
\vspace{-2mm}
\begin{tabbing}%
199 \= xxx\=xxx\=xxx\=xxx\=xxx\=xxx\=xxx\=xxx \+ \kill 
\hspace{-5mm}\textbf{1} $g \leftarrow$ \textsc{True} \\
\hspace{-5mm}\textbf{2} $m \leftarrow \sqrt{n}$ \\
\hspace{-5mm}\textbf{3} \textbf{for} \= $j \leftarrow 1$ \textbf{to} $m$ \\             
\hspace{-5mm}\textbf{4}              \> $c_j \leftarrow 1$\\
\hspace{-5mm}\textbf{5} \textbf{for} \= $i \leftarrow 1$ \textbf{to} $n$ \\
\hspace{-5mm}\textbf{6}              \> $r \leftarrow \lceil s_i/m \rceil $ \\
\hspace{-5mm}\textbf{7}              \>             \> \textbf{for} \= $j \leftarrow 1$ \textbf{to} $c_r - 1$ \\
\hspace{-5mm}\textbf{8}              \>             \>              \> \textbf{if } \ \ \  \= $s_i = Q_{r,j}$ \\ 
\hspace{-5mm}\textbf{9}              \>             \>              \>                     \> $g \leftarrow \textsc{False}$ \\
\hspace{-7mm}\textbf{10}             \>             \>              \>                     \> \textbf{return } $g$ \\
\hspace{-7mm}\textbf{11}             \>             \>  $Q_{r,c_r}\leftarrow s_i$ \\
\hspace{-7mm}\textbf{12}             \>             \>  $c_r \leftarrow c_r + 1$\\
\hspace{-7mm}\textbf{13} \textbf{return} $g$
\end{tabbing}

For the simplicity let us suppose that $m$ is a positive integer and $n = m^2$. 

In the best case $s_1 = s_2.$ Then \textsc{Bucket} executes 1 comparisons in line 8, $m$ assignments in line 4, 
and 1 assignment in line 1, 1 in line 2, 2 in line 6, and 1 in line 8, 11 and 12, therefore 
$T_{best}(n,\textsc{Bucket}) = m + 7 = \Theta(\sqrt{n}).$ The worst case appears, when the input is bad. Then each 
bucket requires $1 + 2 + \cdots + m - 1 = B(n - 1,2)$ comparisons in line 8, further $3m$ assignments in lines 6, 
and 12, totally $\frac{m^2(m - 1)}{2} + 3m^2$ operations. Lines 1, 2, and 9 require 1 assignment per line, and 
the assignment in line 4 is repeated $m$ times. So $T_{worst}(n,\textsc{Bucket}) = 
\frac{m^2(m - 1)}{2} + 3m^2 + m + 3 = \Theta(n^{3/2}).$ 

In connection with the expected behaviour of \textsc{Bucket} 
at first we show that the expected number of elements in a bucket has 
a constant bound which is independent from $n$.

\begin{lemma} Let\label{lemma:14} $b_j(n) = b_j \ (j = 1, \ 2, \ \ldots, \ m)$ be a random variable characterising the number 
of elements in the bucket $B_j$ at the moment of the first repetition. Then
\begin{equation}
\hbox{E}\{b_j\} = \sqrt{\frac{\pi}{2}} - \mu \quad \hbox{for } j = 1, \ 2, \ \ldots, \ m \koz, \label{equation33:69} 
\end{equation}
where 
\begin{equation}
\mu(n) = \mu = \frac{1}{3\sqrt{n}} - \frac{\kappa}{\sqrt{n}} \koz, \label{equation33:70}
\end{equation} 
and $\mu$ tends monotonically decreasing to zero when $n$ tends to infinity.
\end{lemma}
\begin{proof} Due to the symmetry of the buckets it is sufficient to prove 
(\ref{equation33:69}) and (\ref{equation33:70}) for $j = 1.$

Let $m$ be a positive integer and $n = m^2.$ Let $y$ be the random variable defined in (\ref{equation2:41}) and 
$p_k$ be the probability defined in (\ref{equation2:16}).

Let $A_i(n) = A_i \ (i = 1, \ 2, \ \ldots, \ n)$ be the event that  
the number $i$ appears in \textbf{s} before the first repetition and $Y_i(n) = Y_i$ be the indicator of $A_i$. 
Then using the theorem of the full probability we have 
$$
\hbox{E} \{b_1 \} = \sum _{i=1}^m Y _i = \sum _{i=1} ^m \hbox{Pr} \{ A_i\} = m \hbox{Pr} \{A_1 \} \label{equation33:71}
$$
and   
$$
\hbox{Pr} \{ A_1 \} = \hbox{Pr} \{1 \in \{ s_1, \ s_2, \ \ldots, \ s_k \} | y = k \} 
= \sum _{k=1}^n p_k\frac{k}{n} = \frac{1}{n} \sum _{k=1}^n p_k k = \frac{1}{n} R_1 \koz. \label{equation33:72}
$$
Using Lemma \ref{lemma:7}, we get
$$
\hbox{E} \{b_1 \} = m \frac{1}{n} R_1 = \frac{m}{n} \left ( \sqrt{\frac{\pi n}{2}} 
- \frac{1}{3} + \kappa \right) \koz,\label{equation33:73} 
$$
resulting (\ref{equation33:69}) and (\ref{equation33:70}).

We omit the proof of the monotonity of $\mu,$ since it is similar to the corresponding part 
in the proof of Lemma \ref{lemma:7}.
\end{proof}

Table \ref{table:3} shows some concrete values.

\begin{table}[!t] 
\begin{center}
\begin{tabular}{|c|c|c|c|c|c|}    \hline
$n$ &$\hbox{E}\{b_1\}$&$\sqrt{\pi/2}$&$1/(3\sqrt{n})$&$\kappa/\sqrt{n}$ 
&$\mu$\\ \hline
$1$ &$1.000000$       &$1.253314$     &    $0.333333$       & $0.080019$                  &0.253314    \\ \hline
$2$ &$1.060660$       &$1.253314$     &    $0.235702$       & $0.043048$                  &0.192654  \\ \hline
$3$ &$1.090055$       &$1.253314$     &    $0.192450$       & $0.029686$                  &0.162764   \\ \hline
$4$ &$1.109375$       &$1.253314$     &    $0.166667$       & $0.022727$                  &0.143940     \\ \hline
$5$ &$1.122685$       &$1.253314$     &    $0.149071$       & $0.018442$                  &0.130629   \\ \hline
$6$ &$1.132763$       &$1.253314$     &    $0.136083$       & $0.015532$                  &0.120551    \\ \hline
$7$ &$1.140740$       &$1.253314$     &    $0.125988$       & $0.013423$                  &0.112565   \\ \hline
$8$ &$1.147287$       &$1.253314$     &    $0.117851$       & $0.011824$                  &0.106027     \\ \hline
$9$ &$1.152772$       &$1.253314$     &    $0.111111$       & $0.010569$                  &0.100542  \\ \hline
$10$&$1.157462$       &$1.253314$     &    $0.105409$       & $0.009557$                  &0.095852  \\ \hline
\end{tabular}
\caption{Values of $\hbox{E}\{b_1\}$, $\sqrt{\pi/2},$ $1/(3\sqrt{n}),$  
$\kappa/\sqrt{n},$ and $\mu = 1/(3\sqrt{n}) - \kappa/\sqrt{n}$  
for $n = 1, \ 2, \ \ldots, \ 10$ \label{table:3}}
\end{center}
\end{table}

\begin{lemma} Let\label{lemma:15} $f(n) = f$ be a random variable characterising the number of comparisons executed 
in connection with the first repeated element. Then 
$$
\hbox{E}\{f\} = 1 + \sqrt{\frac{\pi}{8}} - \eta \koz, \label{equation33:74}
$$
where 
$$
\eta(n) = \eta = \frac{1/6 + \sqrt{\pi/8} - \kappa/2}{\sqrt{n}+ 1} \koz, \label{equation33:75}
$$
and $\eta$ tends monotonically decreasing to zero when $n$ tends to infinity.
\end{lemma}

\begin{proof} Let $p(i,j,k,n) = p(i,k,n)$ be the probability of the event that there are $k$ different elements before 
the first repetition, and the repeated element belongs to $B_j,$ and $B_j$ contains $i$ elements 
in the moment of the first repetition. Due to the symmetry $p(i,j,k,n)$ does not depend on $j$ and 
$$
p(i,j,k,n) = \binom{m}{i}\binom{n-m}{k-i}k!\frac{i}{n^{k+1}} \koz, \label{equation33:76}
$$
since we investigate $n^{k+1}$ sequences, and if there are $k \ (1 \leq k \leq n)$ different elements before the repeated one, then 
we can choose $i$ elements for the $j$th bucket in $\binom{m}{i}\binom{n - m}{k - 1}$ manner, we can permute them in $k!$ manner, 
and we can choose the repeated element in $i$ manner. Then 
\begin{equation}
\mbox{E}\{f\} = \sum_{i,j,k,n}p(i,j,k)\frac{i+1}{2} - m p_n\label{equation33:77}
\end{equation}
\begin{equation} 
= \frac{m}{2n}\sum_{k=1}^n \frac{k!}{n^k}\sum_{i=1}^k\binom{m}{i}\binom{n-m}{k-i}i(i+1) - p_n\frac{n+1}{2}\label{equation33:78}
\end{equation}

The last member of the formula takes into account that if $k = n,$ then 
additional comparisons with the elements of the bucket corresponding to the repeated element are not necessary. 

Let
$$
\mbox{E'}\{f\}= \mbox{E}\{f\} + p_n\frac{n+1}{2} \koz.\label{equation33:80}
$$

Then dividing the inner sum in (\ref{equation33:78}) by $\binom{n}{k}$ we get the expected value of the random 
variable $\xi(\xi+1),$ where $\xi$ has hypergeometric distribution with parameters $n, \ m,$ and $k.$ It is easy to compute that 
$$
\mbox{E'}\{ \xi(\xi + 1) \} = \mbox{E'} \{ \xi \} (\mbox{E'} \{\xi + 1 \}) + \mbox{Var}\{ \xi \} = 
\frac{km[k(m-1) + (2n-1-m)]}{n(n-1)} \koz,\label{equation33:81}
$$
therefore
\begin{equation}
\mbox{E'}\{ f \} = \frac{m}{2n} \sum_{k=1}^n\frac{k!}{n^k}\binom{n}{k}\frac{km[k(m-1) + (2n-1-m)]}{n(n-1)}\label{equation33:82}
\end{equation}
$$
= \frac{1}{2(n-1)}\sum_{k=1}^n p_k[k(m-1) + (2n-1-m)] \label{equation33:83}
$$
\begin{equation}
= \frac{m-1}{2(n-1)}R_1 + \frac{2n-1-m}{2(n-1)} = \frac{2m+1+R_1}{2m+2} \label{equation:33:84}
\end{equation} 
\begin{equation} = 1 +\sqrt{\frac{\pi}{8}} - \frac{1/6 + \sqrt{\pi/8} 
- \kappa/2}{\sqrt{n}+1} \koz.\label{equation33:85}
\end{equation}
The convergence and monotonicity of $\eta$ is the consequence of the properties of $\kappa.$ Taking into account the small 
value of $p_n$ (see equation (\ref{equation2:16})) the difference $E'\{f\} - E\{f\}$ has negligible influence on the limit of $E\{f\}.$
\end{proof}

\begin{theorem} The expected\label{theorem:16} number of comparisons $C_{exp}(n,\textsc{Bucket}) = C_B$ of \textsc{Bucket} is 
\begin{equation}
C_B = \sqrt{n} + \frac{1}{3} - \sqrt{\frac{\pi}{8}} + \rho \koz,\label{equation33:86}
\end{equation}
where
\begin{equation}
\rho(n) = \rho = \frac{5/6 - \sqrt{9\pi/8} - 3\kappa/2}{\sqrt{n}+1} \koz.\label{equation33:87}
\end{equation}
and $\rho$ tends monotonically decreasing to zero when $n$ tends to infinity.
\end{theorem}
\begin{proof} Let \textbf{s} = $(s_1, \ s_2, \ \ldots, \ s_n)$ be the input sequence of the
algorithm \textsc{Bucket}. \textsc{Bucket} processes the input sequence using 
$m = \sqrt{n}$ buckets $B_1, \ B_2,$ \linebreak 
\noindent $\ldots, B_n$: it investigates the input elements sequentially and 
if the $i$-th input element $s_i$ belongs to the interval $[(r - 1)m + 1, (r - 1)m +2, \ldots, \ rm]$, then it 
sequentially compares $s_i$ with the elements in the bucket $B_r$ and finishes, if it finds a collision, or puts $s_i$ 
into $B_r$, if $s_i$ differs from all elements in $B_r$.    

Let $y$ be the random variable, defined in (\ref{equation31:43}), and $p_k$ the probability defined in (\ref{equation2:16}).
Let $b_i$ be the random variable defined in Lemma \ref{lemma:14}, and $c_j(n) = c_j \ (j = 1, \ 2, \ \ldots, \ m)$ be a 
random variable characterising the number of comparisons executed in $B_j$ before the processing of the first repeated element, 
and $c(n) = c$ a random variable characterising the number of necessary comparisons executed totally by \textsc{Bucket}. 
Then due to the symmetry we have 
\begin{equation}
C_B = \hbox{E} \left \{ \sum_{j=1}^m c_j \right \} + E\{ f \} 
=  m\hbox{E} \{c_1\} + E\{ f \} \koz. \label{equation33:88}
\end{equation}

The probability of the event $A(i_1,i_2,k,n) = A(i_1,i_2,k)$ that the elements $i_1$ and 
$i_2 \ (1 \leq i_1, \ i_2 \leq m)$ will be compared 
before the processing of the first repeated element at the condition that $y = k$ and $2 \leq k \leq n$ equals to 
$$
\hbox{Pr} \{ A(i_1,i_2,k) | y = k \mbox{ and } 2 \leq k \leq n \} = 
\frac{\binom{n-2}{k-2}}{\binom{n}{k}} = \frac{k(k - 1)}{n(n - 1)} \koz ,\label{equation33:89}
$$

Since there are $\binom{m}{n}$ possible comparisons among the elements of the interval $[1,m],$ we have
$$
E\{ c_1 \} = \sum _{k=1}^{n} p_k \frac{k(k-1)}{n(n-1)} \binom{m}{2} 
= \frac{m(m-1)}{2n(n-1)}\left(\sum _{k=1}^{n} p_k k^2 - \sum _{k=1}^{n} p_k k \right) \koz,\label{equation33:90} 
$$
from where using Lemma \ref{lemma:7} and Lemma \ref{lemma:8} we get
\begin{equation}
E\{ c_1 \} = \frac{n - \sqrt{n}}{2n^2 -2n} (R_2 - R_1) = 
\frac{1}{2n + 2\sqrt{n}} \left[2n - 2\left(\sqrt{\frac{\pi n}{2}} - \frac{1}{3} + \kappa \right)\right] \koz.\label{equation33:91}
\end{equation}

This equality implies
\begin{equation}
E\{ c_1 \} = 1 - \frac{1}{\sqrt{n} + 1} \left(\sqrt{\frac{\pi}{8}} + \frac{2}{3} - \kappa \right) \koz.\label{equation33:92}
\end{equation}

From (\ref{equation33:88}), taking into account (\ref{equation33:92}), (\ref{equation33:82}), and (\ref{equation33:85})  we get 
$$
C_B = \sqrt{n} + \frac{1}{3} - \sqrt{\frac{\pi}{8}} +  \frac{\sqrt{9\pi/8}  
+ 5/6 - 3\kappa/2}{\sqrt{n}+1} \koz .\label{equation33:93}
$$
Denoting the last fraction by $\rho$ we get the required (\ref{equation33:86}). The monotonity of $\rho$ 
is the consequence of the monotonity of $\kappa$.
\end{proof}

\begin{theorem} The expected\label{theorem:17} running time $T_{exp}(n,\textsc{Bucket}) = T_B$ of \textsc{Bucket} is 
\begin{equation}
T_B = \sqrt{n}\left (3 + 3 \sqrt{\frac{\pi}{2}} \right) + \sqrt\frac{25\pi}{8} + \phi \koz, \label{equation33:94}
\end{equation}
where 
$$
\phi(n) = \phi = 3\kappa - \rho - 3 \eta - \frac{n!}{n^n} - \frac{3\sqrt{\pi/8} - 1/3  
- 3\kappa/2}{\sqrt{n} + 1} \koz, \label{equation33:95}
$$
and $\phi$ tends to zero when $n$ tends to infinity.
\end{theorem}
\begin{proof}
\textsc{Bucket} requires 2 assignments in lines 1 and 2, $\sqrt{n}$ assignments in line 4, $R_1$ assignments in line 
6, $C_B + E\{ f \}$ assignments in line 8, $1 - p_n$ expected assignment in line 9 and $2R_1$ assignments in lines 
11 and 12 before the first repeated element, and $2\mbox{E}\{ f \} - 1$ assignments after the first repeated element. 

Therefore the expected number $A_{exp}(n,\textsc{Bucket}) = A_B$ of assignments of \textsc{Bucket} is 
$$
A_B =  2 + \sqrt{n} + 3R_1 + C_B + 3 E\{f \} - \frac{n!}{n^n} \koz .\label{equation33:96}
$$
Substituting $R_1,$ and $C_B,$ and $E\{f\}$ we get
\begin{equation}
A_B = 2 \sqrt{n} + \frac{13}{3} + 3 \sqrt{\frac{\pi n}{2}} + 3 \kappa   
- \sqrt{\frac{\pi}{8}} + \rho + 3 \sqrt{\frac{\pi}{8}} - 3 \eta  - \frac{n!}{n^n} \koz , \label{equation33:97}
\end{equation}
implying
$$
A_B = \sqrt{n} \left ( 2 + 3 \sqrt{\frac{\pi}{2}} \right ) + \frac{13}{3} + \sqrt{\frac{\pi}{2}} + 3\kappa + \rho - 3 \eta  
- \frac{n!}{n^n} \koz  . \label{equation33:98}
$$

Summing up the expected number of comparisons in (\ref{equation33:86}) and of assignments in (\ref{equation33:97}) 
we get the final formula (\ref{equation33:94}).
\end{proof}

\subsection{Test of random arrays \label{subsection:34}}
\textsc{Matrix} is based on \textsc{Bucket}.

For the simplicity let us suppose that $n$ is a square.

Let $\mathcal{M}$ be an $n \times n$ sized matrix, where $m_{ij} \in \{1,2,\ldots,n\}$. 
The $i$th row of $\mathcal{M}$ is denoted by $r_i$, and the $j$th column  
by $c_j$ for $1 \leq i, j \leq n$. The matrix $M$ is called \textit{good,} if its 
all lines (rows and columns) contain a permutation of the elements $1, \ 2, \ \ldots, \ n$.

\newpage
\noindent \textsc{Matrix}$(n,\mathcal{M})$
\vspace{-2mm}
\begin{tabbing}%
199 \= xxx\=xxx\=xxx\=xxx\=xxx\=xxx\=xxx\=xxx \+ \kill 
\hspace{-7mm}\textbf{1} $g \leftarrow$ \textsc{True} \\
\hspace{-7mm}\textbf{2} \textsc{Bucket}$(n,r_1)$ \\
\hspace{-7mm}\textbf{3} \textbf{if} \= $g = \textsc{False}$ \\             
\hspace{-7mm}\textbf{4}             \> \textbf{return} $g$\\
\hspace{-7mm}\textbf{5} \textbf{for} \= $i \leftarrow 2$ \textbf{to} $n$ \\
\hspace{-7mm}\textbf{6}              \> \textsc{Bucket}$(n,r_i)$ \\
\hspace{-7mm}\textbf{7}              \> \textbf{if} \= $g = \textsc{False}$ \\
\hspace{-7mm}\textbf{8}              \>             \> \textbf{return} $g$ \\
\hspace{-7mm}\textbf{9} \textbf{for} \= $j \leftarrow 1$ \textbf{to} $n$ \\
\hspace{-7mm}\textbf{10}             \> \textsc{Bucket}$(n,c_j)$ \\
\hspace{-7mm}\textbf{11}             \> \textbf{if} \= $g = \textsc{False}$ \\
\hspace{-7mm}\textbf{12}             \>             \> \textbf{return} $g$ \\ 
\hspace{-7mm}\textbf{13} \textbf{return} $g$
\end{tabbing}

\begin{theorem} The\label{theorem:18} expected running time $T_{exp}(n,\textsc{Matrix}) = T_M$ of \textsc{Matrix} is 
\begin{equation}T_M = T_B + o(1) \koz. \label{equation34:99}
\end{equation}
\end{theorem}
\begin{proof} According to Theorem \ref{theorem:17} we have
$$
T_B = \sqrt{n}\left (3 + 3 \sqrt{\frac{\pi}{2}} \right) + \sqrt\frac{25\pi}{8} + o(1) \koz.\label{equation34:100}
$$
Since the rows of $\mathcal{M}$ are independent, therefore the probability of the event $G_k(n) = G_k \ (k = 1, \ 2, \ldots, n)$ 
that the first $k$ rows are good is
$$
\hbox{Pr} \{G_k \} = \left ( \frac{n!}{n^n} \right )^k \koz,\label{equation4:101}
$$
so for the expected time $T_{exp}(n,\textsc{Matrix}) = T_R$ of the testing of the rows we have
$$
T_R \leq T_B + T_B\sum _{k=1}^{n-1} \left (\frac{n!}{n^n} \right )^k = T_B + o(1) \koz.\label{equation34:102}
$$

Since the columns are also independent, all the rows and the first $k$ columns are good with the probability
$$
p = \left ( \frac{n!}{n^n} \right )^{n + k} \koz, \label{equation34:103}
$$
and so for the expected time of testing of the columns $T_{exp}(n,\textsc{Matrix}) = T_C$ holds
$$
T_C \leq T_B \sum _{k=0}^{n-1} \left ( \frac{n!}{n^n} \right ) = o(1) \koz, \label{equation34:104}
$$
and so
$$
T_M = T_R + T_C\label{equation34:105}
$$
implies (\ref{equation34:99}).
\end{proof}

\section{Summary \label{section:4}}
Table \ref{table:4} summarises the basic properties of the number of necessary comparisons of the investigated algorithms.

\begin{table}[!t] 
\begin{center}
\begin{tabular}{|l|c|c|c|c|c|c|c|c|}    \hline
Index and algorithm   &  $C_{best}(n)$        &  $C_{worst}(n)$      & $C_{exp}(n)$  \\ \hline
1. \textsc{Linear}    &   $\Theta(1)$         & $\Theta(n)$          & $\Theta(\sqrt{n})$         \\ \hline
2. \textsc{Backward}  &   $\Theta(1)$         & $\Theta(n^2)$        & $\Theta(n)$          \\ \hline
3. \textsc{Bucket}    &   $\Theta(1)$         & $\Theta(n \sqrt{n})$ & $\Theta(\sqrt{n})$ \\ \hline
4. \textsc{Matrix}    &   $\Theta(1)$         & $\Theta(n \sqrt{n})$ & $\Theta(\sqrt {n})$ \\ \hline
\end{tabular}
\caption{The expected number of comparisons of the investigated algorithms in best, worst and expected cases \label{table:4}}
\end{center}
\end{table}

Table \ref{table:5} summarises the basic properties of the running times of the investigated algorithms.

\begin{table}[!t] 
\begin{center}
\begin{tabular}{|l|c|c|c|c|c|c|c|c|}    \hline
Index and algorithm   &  $T_{best}(n)$   &  $T_{worst}(n)$ & $T_{exp}(n)$  \\ \hline
1. \textsc{Linear}    &   $\Theta(n)$           &       $\Theta(n)$      & $n + \Theta(\sqrt{n})$         \\ \hline
2. \textsc{Backward}  &   $\Theta(1)$           &       $\Theta(n^2)$    & $\Theta(n)$          \\ \hline
3. \textsc{Bucket}    &   $\Theta(\sqrt{n})$    &   $\Theta(n \sqrt{n})$ &   $\Theta(\sqrt {n})$ \\ \hline
4. \textsc{Matrix}    &   $\Theta(\sqrt{n})$           &   $\Theta(n \sqrt{n})$ &   $\Theta(\sqrt {n})$ \\ \hline
\end{tabular}
\caption{The running times of the investigated algorithms in best, worst and expected cases \label{table:5}}
\end{center}
\end{table}

We used in our calculations the RAM computation model \cite{CormenLe2009}. If the investigated 
algorithms run on real computers then we have to take into account also the limited capacity of the memory 
locations and the increasing execution time of the elementary arithmetical and logical operations. 

\section*{Acknowledgements} Authors thank Tam\'as F. M\'ori \cite{Mori2010} for proving Lemmas \ref{lemma:14} 
and  \ref{lemma:15}, P\'eter Burcsi \cite{Burcsi2009} for useful information 
on references (both are teachers of E\"otv\"os Lor\'and University) and the unknown referee for the 
useful corrections. \\

The European Union and the European Social Fund have provided financial support to the project 
under the grant agreement no. T\'AMOP 4.2.1/B-09/1/KMR-2010-0003.

\bigskip
\rightline{\emph{Received:  January 11, 2011 {\tiny \raisebox{2pt}{$\bullet$\!}} Revised: April 5, 2011}} 

\end{document}